\documentclass[aps,12pt,tightenlines,prd,onecolumn,superscriptaddress,eqsecnum, nofootinbib]{revtex4}
\usepackage{graphicx}
\usepackage{dcolumn}
\usepackage[dvipsnames]{xcolor}
\usepackage{bm}
\usepackage{amssymb}
\usepackage{amsthm}
\usepackage[colorlinks=true,linkcolor=blue,citecolor=blue,urlcolor=blue]{hyperref}
\usepackage{amsmath}
\usepackage{tikz}
\usetikzlibrary{decorations.pathmorphing}
\usepackage{subfig}
\usepackage{enumerate}   
\usepackage{datetime}

\newcommand{\be}{\begin{equation}}
\newcommand{\ee}{\end{equation}}
\newcommand{\bea}{\begin{eqnarray}}
\newcommand{\eea}{\end{eqnarray}}

\newcommand{\dd}{{\rm d}}
\newcommand{\x}{{\rm x}}
\newcommand{\ii}{{\rm i}}
\newcommand{\expe}{{\rm e}}
\newcommand{\unit}{1\!\!1}

\usepackage{mathrsfs}

\newcommand{\gre}{\epsilon}

\newcommand{\grl}{\lambda}

\newcommand{\grr}{\rho}
\newcommand{\grs}{\sigma}
\newcommand{\grt}{\tau}

\newcommand{\grw}{\omega}

\newcommand{\nn}{\nonumber}

\newtheorem{theorem}{Theorem}
\newtheorem{lemma}[theorem]{Lemma}
\newtheorem{proposition}[theorem]{Proposition}

\theoremstyle{definition}
\newtheorem{defn}[theorem]{Definition}

\newdateformat{daymonthyear}{\THEDAY \, \monthname[\THEMONTH] \THEYEAR}
\newdateformat{monthyear}{\monthname[\THEMONTH] \THEYEAR}

\begin{document}

\title{Asymptotic states for stationary Unruh-DeWitt detectors}
\author{Benito A. Ju\'arez-Aubry}\email{benito.juarez@iimas.unam.mx}
\affiliation{Departamento de F\'isica Matem\'atica, \\
Instituto de Investigaciones en Matem\'aticas Aplicadas y en Sistemas, \\ Universidad Nacional Aut\'onoma de M\'exico,\\Apartado Postal 20-126, 01000 Mexico City, Mexico}
\author{Dimitris Moustos}\email{dmoustos@upatras.gr}
\affiliation{Department of Physics, University of Patras, 26504 Patras, Greece}
\date{\daymonthyear\today}

\begin{abstract}

We study the late-time asymptotic state of a stationary Unruh-DeWitt detector interacting with a field in a thermal state. We work in an open system framework, where the field plays the role of an environment for the detector. The long-time interaction between the detector and the field is modelled with the aid of a one-parameter family of switching functions that turn on and off the interaction Hamiltonian between the two subsystems, such that the long-time interaction limit is reached as the family parameter goes to infinity. In such limit, we show that if the field is in a Kubo-Martin-Schwinger (KMS) state and the detector is stationary with respect to the notion of positive frequency of the field, in the Born-Markov approximation, the asymptotic state of the detector is a Gibbs state at the KMS temperature. We then relax the KMS condition for the field state, and require only that a frequency-dependent version of the detailed balance condition for the Wightman function pulled back to the detector worldline hold, in the sense that the inverse temperature appearing in the detailed balance relation need not be constant. In this setting, we show that the late-time asymptotic state of the detector has the form of a thermal density matrix, but with a frequency-dependent temperature. We present examples of these results, which include the classical Unruh effect and idealised Hawking radiation  (for fields in the HHI state), and also the study of the late-time behaviour of detectors following stationary ``cusped" and circular motions in Minkowski space interacting with a massless Klein-Gordon field in the Minkowski vacuum. In the cusped motion case, a frequency-dependent, effective temperature for the asymptotic late-time detector state is obtained analytically. In the circular motion case, such effective temperature is obtained numerically.
\end{abstract}

\maketitle

\section{Introduction}


An Unruh-DeWitt detector \cite{Unruh,Dewitt} (see also \cite{Birrell}) is a microscopic system coupled with a quantum field through an interaction Hamiltonian, which seeks to operationally probe the physical properties of the field. Information about the field is incorporated in the final state of the detector and can be extracted by performing a suitable measurement. A notable application of Unruh-DeWitt detectors is to probe the thermal properties of quantum fields in the case of fields propagating in curved spacetime backgrounds or detectors following non-inertial trajectories in Minkowski spacetimes. An example of the former case is (i) the Hawking effect \cite{hawking}, according to which a black hole radiates to infinity its mass in the form of thermal radiation, which in the case of the Schwarzschild spacetime is at a temperature $T=\kappa/(2\pi)$, where $\kappa$ is the surface gravity of the black hole. The latter case is (ii) the Unruh effect \cite{Unruh}, according to which an observer moving with uniform proper acceleration $a$, experiences the  Minkowski vacuum as a heat bath at the Unruh temperature $T_U = a/(2\pi)$.

The aim of this work is to study the emergence of thermality in Unruh-DeWitt detectors that follow stationary trajectories. To this end, we consider a two-level system detector that is linearly coupled to a Klein-Gordon field in a Kubo-Martin-Schwinger (KMS) state \cite{Kubo, Schwinger, Haag}. We use a real smooth compactly supported function to specify how the interaction between the detector and the field is switched on and off. In this framework, we treat the detector as an open quantum system, with the field playing the role of the environment, and present a systematic way to calculate the asymptotic state of the detector after a long interaction with the field, employing the mathematical properties of the one-parameter family of adiabatically-scaled switching functions, where the parameter controls the interaction time between the detector and the field. Employing the Born-Markov approximation, we find that if the detailed balanced condition is satisfied for the pullback of the Wightman function of the field along the worldline of the detector, then the detector always thermalises to an equilibrium Gibbs state at the KMS temperature.

The open system formalism that we adopt (see, e.g., \cite{Breuer}) also allows us to conclude that the KMS condition is not the most general condition on a field state whereby a detector relaxes to a stationary state. To wit, if a {\it frequency-dependent} version of the detailed balance condition holds for the Wightman function along the detector's worldline, in which the temperature is a function of the detector's frequency, the detector will relax to a state to which a frequency-dependent effective temperature can be associated. As we shall see, this result opens, for example, the possibility of associating effective temperatures to stationary trajectories in Minkowski spacetime other than the linearly uniformly accelerated one, such as the ones reported in \cite{Letaw}.

We give a few examples of applications of our result. In particular, we consider the Unruh effect for uniformly accelerated detectors and the Hawking effect for detectors interacting with a field in the Hartle-Hawking-Israel vacuum. Furthermore, we consider the cases of cusped and circular trajectories followed by the detector, where we find that the detector thermalises at a frequency-dependent effective temperature. In the cusped motion case, such effective temperature can be computed analytically, while in the circular case we resort to numerical computations.

The structure of the article is the following. In Sec. \ref{UDW:KMS}, we introduce the Unruh-DeWitt detector model and give a brief review of the basic properties of the KMS states, and its relation to the detailed balance condition \cite{BL}. In Sec. \ref{evol:eq}, we derive, in the Born-Markov approximation, the time evolution equations for the reduced density matrix of an Unruh-DeWitt detector interacting with a field in a thermal state. In Sec. \ref{Sec:Equilibrium}, we study the thermalisation of the detector in the long-time interaction limit at a constant or frequency-dependent temperature. In Sec. \ref{apps}, we apply the results obtained in the previous sections to different stationary trajectories of the detector that lead to its thermalisation at a constant or frequency-dependent temperature. Finally, in Sec. \ref{concl}, we summarise and discuss our results, and give some perspectives.


\subsection{Notation and preliminaries}
\label{Notation}

We use the following conventions: A spacetime $\mathcal{M}:= (M,g)$ consists of a real four-dimensional, (connected, Hausdoff, paracompact) smooth manifold, $M$, equipped with a Lorentzian metric $g$. Spacetime points are denoted by Roman characters ($\x$). We restrict our spacetimes to be globally hyperbolic and deal throughout only with static spacetimes. We fix $c = \hbar = k  = 1$. Complex conjugation is denoted by an overline. The adjoint of a Hilbert-space operator, $\hat o$, is denoted by $\hat o^*$. We use the initials H.c. to denote Hermitian conjugate. $O(x)$ denotes a quantity for which $O(x)/x$ is bounded as $x \to 0$. We use the standard notation $C_0^\infty(M)$ for the space smooth functions of compact support on $M$ with distributional dual $\mathscr{D}'(M)$. $\mathscr{S}(\mathbb{R})$ denotes the space of Schwartz functions on $\mathbb{R}$ and $\mathscr{S}'(\mathbb{R})$ is the space of tempered distributions. The convolution between two functions $f,g: \mathbb{R} \to \mathbb{R}$ is denoted by $(f \star g)(t) = \int_\mathbb{R} \! \dd s f(s) g(t-s)$. Our Fourier transform convention is $\widehat f(\omega) = \int_\mathbb{R} \! \dd t \, \expe^{-\ii \omega t} f(t)$ and for the inverse Fourier transform as $g(t) = (2 \pi)^{-1} \int_\mathbb{R} \! \dd \omega \, \expe^{\ii \omega t} \widehat f(\omega)$, where we used wide carets to denote the Fourier transform, $\widehat{f} = \mathcal{F}[f]$. 

\section{The Unruh-DeWitt detector model}\label{UDW:KMS}

\subsection{The coupling between the detector and the field}

We consider static spacetimes, i.e., spacetimes of the form $(\mathbb{R}\times \Sigma, g)$ where the spacetime has an isometry generated by an irrotational timelike Killing vector field. In other words, the metric tensor admits the form $g(t,x) = \beta(x) dt^2 +  h(x)$, where $h$ is a Riemannian metric on $\Sigma$ and $x$ are coordinates on $\Sigma$, and with $\partial_t$ being the relevant Killing vector field.

In such static spacetimes, we model an Unruh-DeWitt detector by a pointlike two-level quantum system with frequency gap $\grw>0$, moving through spacetime and interacting with a quantum field.

More precisely, the detector is a two-level system with spacetime trajectory $\x(\tau)$, with $\tau$ the proper time of the detector, described by a monopole moment operator expressed in terms of the ladder operators $\hat{\grs}_{\pm}$ of the Pauli algebra as
\be\label{mmoment}
\hat \mu(\grt)=e^{\ii\grw\grt}\hat{\grs}_++e^{-\ii\grw\grt}\hat{\grs}_-,
\ee
which acts on the detector Hilbert space, $\mathscr{H}_{\rm D} \simeq \mathbb{C}^2$, spanned by the orthonormal basis $\{|0\rangle, |1 \rangle \}$. The Hamiltonian of the detector with respect to $\tau$, $\hat H_{\rm D} =\omega\hat{\sigma}_3/2$, has the elements of the orthonormal basis as frequency (energy) eigenstates $\hat H_{\rm D} |0\rangle = 0 |0\rangle$ and $\hat H_{\rm D} |1\rangle = \omega |1 \rangle$.

We shall consider throughout that the Unruh-DeWitt detector interacts with a real Klein-Gordon field, i.e., an operator-valued distribution $f \mapsto \hat \Phi(f)$, for $f \in C_0^\infty(M)$ (or a suitable test-function space), acting densely on a Hilbert space, $\hat \Phi(f): \mathscr{D} \subset \mathscr{H} \to \mathscr{H}$, such that (i) the map $f \mapsto \hat \Phi(f)$ is linear, (ii) the field is self-adjoint $\hat \Phi(f)^* = \hat \Phi(f)$, (iii) the Klein-Gordon equation is obeyed $\hat \Phi\left( (\Box -m^2 - \xi R) f \right) = 0$, which is equivalent to $(\Box_\x -m^2 - \xi R(\x)) \hat \Phi(\x)$ by integration by parts and the arbitrariness of the test function $f$ and (iv) the commutation relations are fulfilled, i.e., for $f_1, f_2 \in C_0^\infty(M)$ we have $[\hat \Phi(f_1), \hat \Phi(f_2)] = - \ii E(f_1, f_2)$, where $E$ is the advanced-minus-retarded Green operator of $\Box -m^2 - \xi R$. The Hamiltonian of the theory can be unambiguously defined with respect to the timelike Killing vector of spacetime.

In static spacetimes, the unsmeared quantum field can be written explicitly, as follows. The field equation can be written explicitly as $-\partial_t^2 \hat \Phi = g_{00} (\Delta_x - m^2 - \xi R) \hat \Phi =: K \hat \Phi$, which for harmonically time-dependent solutions becomes the eigenvalue problem on $\Sigma$,
\begin{equation}
K \psi_j = \omega^2_j \psi_j.
\label{Eigen}
\end{equation}

Assuming that $K$ defines a positive, self-adjoint operator, we have that the quantum field can be written as an eigenfunction expansion of the form
\begin{equation}
\hat{\Phi}(t, x) = \int \! \frac{d\mu(j)}{\sqrt{2 \omega_j}} \left( \hat{a}_j \, \expe^{-\ii \omega_j t} \psi_j(x) + \hat{a}^*_j \, \expe^{+\ii \omega_j t} \overline{\psi_j(x)}  \right).
\end{equation}

A natural one-particle structure, $(\mathcal{H}, K)$, consists of the one-particle Hilbert space, $\mathcal{H}$, consisting of complex, positive-frequency (with respect to the time $t$) solutions equipped with the $L^2(\Sigma, \dd\mu = g_{00}^{-1/2}{\rm dvol}_h)$ inner product and the map $K$ taking classical solutions of the Klein-Gordon equation, $\varphi$, to $K\varphi \in \mathcal{H}$, with dense range in $\mathcal{H}$. A natural zero-temperature vacuum is that for which $\hat{a}(K\varphi) \Omega_\infty = 0$ and for which $\hat{a}^*(K \varphi)$ defines a creation operator. The bosonic Fock space of the theory based on the vacuum is given by $\mathscr{H} = \mathcal{F}(\mathcal{H})$, constructed in the usual way. In the sequel, however, we will be concerned with states at finite temperature $T =1/\beta$, which we denote by $\Omega_\beta$, satisfying the KMS condition, and which are cyclic in the Fock space $\mathscr{H}_\beta$. As we shall see below, the two-point function in such KMS states can also be represented in terms of the modes defined by Eq. \eqref{Eigen}.

The Hamiltonian of the total system is given by $\hat H = \hat H_{\rm D} \otimes \unit+\unit\otimes \hat H_{\hat \Phi}+\hat H_{{\rm int}}$, where $\hat H_{\hat \Phi}$ is the Hamiltonian of the Klein-Gordon field and $\hat H_{{\rm int}}$ is the interaction Hamiltonian. We take the interaction Hamiltonian to be
\be
\hat H_{\text{int}}(\grt)=c\chi(\grt)\hat \mu(\tau)\otimes\hat \Phi\left(\x(\tau)\right),
\label{Hint}
\ee
where $c$ is small a coupling constant
and $\chi\in C_0^{\infty}(\mathbb{R})$ is a real smooth compactly supported function, which specifies how the interaction is switched on and off between the detector's monopole moment operator, $\hat \mu$, and the pullback of the field to the detector's worldline, $\hat \Phi\left(\x(\tau)\right)$.

 Throughout this work, we will be chiefly concerned with the cases in which the detector's worldline, $\x(\tau)$ coincides with the orbit generated by a timelike Killing field, i.e., when the detector follows a {\bf stationary trajectory}. In these cases, when the two-point function of the field is pulled back to the detector trajectory, it will only depend on the points along the worldline as the difference $\tau - \tau'$, as shown by Letaw  \cite{Letaw}.

\subsection{KMS states}
\label{sec:KMS}

Throughout this work, we are interested in the situations in which the field is initially in a KMS state. Such states are defined by a condition on the two-point function. Because they are quasi-free, like the vacuum state, the two-point function contains sufficient information for defining all $n$-point functions, and hence define the state via the Gelfand-Naimark-Segal (GNS) construction. 

Let us start by explicitly giving an expression for the Wightman function of a thermal state at temperature $1/\beta$ in static spacetimes as a mixture of positive and negative frequency modes,
\begin{align}
\mathcal{W}^+_\beta((t,x),(t',x')) & =\langle \Omega_\beta | \hat \Phi(t,x) \hat{\Phi}(t,x') | \Omega_\beta \rangle  \nonumber \\
&  = \int \! \frac{\dd \mu(j)}{2 \omega_j} \frac{\psi_j(x) \overline{\psi_j(x')}}{1 - \expe^{-\beta \omega_j}} \left( \expe^{-\ii \omega_j(t-t')} + \expe^{-\beta \omega_j} \expe^{\ii \omega_j(t-t')}  \right).
\label{WigthBeta}
\end{align}

Equation \eqref{WigthBeta} should be interpreted as a distributional kernel, with $\mathcal{W}^+_\beta$ mapping $C_0^\infty(M) \times C_0^\infty(M) \ni (f,g) \mapsto \mathcal{W}^+_\beta(f,g)$. The thermal Wightman function defined by Eq. \eqref{WigthBeta} is translation invariant, whereby $\mathcal{W}^+_\beta((t,x),(t',x')) = \mathcal{W}^+_\beta((t-t',x),(0,x'))$, and we can associate to it the complex-time two-point function $\mathscr{W}_\beta^+(z, x, x')$, by replacing the time difference $t-t' \in \mathbb{R}$ in Eq. \eqref{WigthBeta} by the variable $z \in \mathbb{C}$. $\mathscr{W}_\beta^+(z, x, x')$ is analytic in the strip $- \beta < \Im z < 0$. 

Similarly, a complex-time two-point function analytic in the strip $0 < \Im z < \beta$, denoted by $\mathscr{W}_\beta^-(z, x, x')$, can be associated to $\mathcal{W}^-_\beta((t,x),(t',x'))$. 

For spacelike separated points, $\mathcal{W}^+_\beta((t,x),(t',x')) = \mathcal{W}^-_\beta((t,x),(t',x'))$ by the commutation relations, and hence  $\mathscr{W}_\beta^+(z, x, x') = \mathscr{W}_\beta^-(z, x, x')$ for $|z| < d(x,x')$ along the real axis, where $d$ denotes the spacelike geodesic distance between the points $x$ and $x'$. Hence, $\mathcal{W}^+_\beta$ and $\mathcal{W}^-_\beta$ are analytic continuations of each other. Generically, there will be branch cuts extending from $(-\infty, -d)$ and $(d, \infty)$ along the real axis on the complex plane, over which the two-point functions ``jump", corresponding to causally separated arguments of the (real-time) two-point functions.

The above procedure can be used to analytically continue the complex-valued functions $\mathscr{W}_\beta^+((t',x'),(t,x))$ at $\Im z = -\beta$ and $\mathscr{W}_\beta^-((t',x'),(t,x))$ at $\Im z = +\beta$, and again at all values $\Im z = n \beta$, $n \in \mathbb{Z}$, where branch cuts are located, with the aid of the relation $\mathscr{W}_\beta^+(z - \ii \beta + \ii \epsilon, x, x') = \mathscr{W}_\beta^-(z+\ii \epsilon, x, x')$, as $\epsilon \to 0^+$, which is more commonly written when no notational confusion may arise as
\begin{equation}
\mathcal{W}_\beta^+((t-t') - \ii \beta + \ii \epsilon, x, x') = \mathcal{W}_\beta^+(-(t-t')-\ii \epsilon, x', x),
\label{KMSCondition}
\end{equation}
where the $\ii \epsilon$ is interpreted in the distributional sense.\footnote{In other words, we mean that for fixed $x, x'$ and for all $f \in \mathscr{S}'(\mathbb{R})$, the weak limits coincide as $\int \! \dd s \, f(s) \mathcal{W}_\beta^+(s - \ii \beta + \ii \epsilon, x, x') = \int \! \dd s \, f(s) \mathcal{W}_\beta^+(-s-\ii \epsilon, x', x)$.} Equation \eqref{KMSCondition} is known as the {\bf KMS condition} and it is the relation that defines thermal states through the analytic structure of their (complex-time) two-point function.

In the case of detectors, in order to compute their response, the two-point function is pulled back to the worldline of the detector and the branch cuts separate strips of width $\beta$ for the pullback of the two-point function, which we denote simply by $\mathcal{W}_\beta^+(\tau - \tau') = \mathcal{W}_\beta^+((t(\tau)-t(\tau')), x(\tau), x(\tau'))$. Based on this observation, the definition of the KMS condition along the worldline was given in \cite[Def. 4.1]{BL}. We restate the definition in order to make the present section self-contained.

\defn[Worldline KMS condition, 4.1 in \cite{BL}]
{Suppose there is a positive constant $\beta$ and a holomorphic
function $\mathcal{W}_\mathbb{C}$ in the strip $S = \{ z \in \mathbb{C} ∣ - \beta < \Im z < 0 \}$ such that
\begin{enumerate}
\item $\mathcal{W}$ is the boundary value of $\mathcal{W}_\mathbb{C}$ on the real axis;
\item $\mathcal{W}_\mathbb{C}$ has a distributional boundary value on the line $\Im s = - \ii \beta$;
\item The two boundary values of $\mathcal{W}_\mathbb{C}$ are linked by
\begin{equation}
\mathcal{W}_\mathbb{C} (s - \ii\epsilon) = \mathcal{W}_\mathbb{C} ( - s - \ii \beta + \ii\epsilon )
\end{equation}
 as $\epsilon \to 0^+$ in the weak limit sense for $s \in \mathbb{R}$.
\item  For any $0 < a < b < \beta$ there is a polynomial $P$, so that $| \mathcal{W}_\mathbb{C}(s) | < P(| \Re s |)$ for all $s \in S$, with $- b < \Im s < - a$.
\end{enumerate}
Then we say that $\mathcal{W}_\mathbb{C}$ obeys the KMS condition at temperature $1/\beta$.}\label{KMSdef} 

We have defined the KMS condition in terms of an abstract holomorphic function $\mathcal{W}_\mathbb{C}$, but in the applications that follow $\mathcal{W}_\mathbb{C}(z)$ will take the concrete form of the complex-time two-point function associated to the worldline-pullback two-point function $\mathcal{W}_\beta^+(\tau-\tau')$. 

Condition 4 in Def. \ref{KMSdef} is a technical requirement, which in ref. \cite{BL} is used to show that the {\bf detailed balance condition} and the KMS condition are equivalent along the worldline. Namely, the following proposition holds:

{\proposition[4.3 in \cite{BL}] \label{PropDBC} Let $\widehat{\mathcal{W}}_\mathbb{C} \geq 0$ be the Fourier transform of $\mathcal{W}_\mathbb{C}$ and let $\widehat{\mathcal{W}}_\mathbb{C}(\omega) \leq P(\omega)$, where $P$ is a Polynomial function. $\mathcal{W}_\mathbb{C}$ satisfies the KMS condition at temperature $1/\beta$ iff $\widehat{\mathcal{W}}_\mathbb{C}$ satisfies the detailed balance condition,
\begin{equation}
\widehat{\mathcal{W}}_\mathbb{C}(-\omega) = \expe^{\beta \omega} \widehat{\mathcal{W}}_\mathbb{C}(\omega).
\label{DBC}
\end{equation}}


We should mention that whenever $\mathcal{W}_\beta^+$ stems from a two-point function, then it must be a distribution of positive type, i.e., as an element of $\mathscr{D}'(\mathbb{R})$, for $f \in C_0^\infty(\mathbb{R})$, it holds that $\mathcal{W}_\beta^+(f \star \overline{\tilde{f}}) \geq 0$, where $\tilde{f}(s) = f(-s)$. By the Bochner-Schwartz theorem, $\mathcal{W}_\mathbb{C}$ is a tempered distribution and $\widehat{\mathcal{W}}_\mathbb{C}$ is a polynomially bounded positive measure. Thus, demanding that $\widehat{\mathcal{W}}_\mathbb{C}$ be a polynomially bounded non-negative function is a weak assumption.


\section{Time evolution of the detector in the Born-Markov approximation}\label{evol:eq}

We treat the detector as an open quantum system, with the Klein-Gordon field playing the role of the environment, inducing dissipation and decoherence. The dynamics of the density operator of the total system in the interaction picture, $\hat \grr_{\text{tot}}: \mathbb{C}^2 \otimes \mathscr{H}_\beta \to \mathbb{C}^2 \otimes \mathscr{H}_\beta$,  is ruled by \cite{Breuer}
\be\label{von}
\frac{d}{d\grt}\hat \grr_{\text{tot}}(\grt)=-\ii\left[\hat H_\text{int}(\grt),\hat\grr_{\text{tot}}(\grt)\right].
\ee

Given that we have assumed that the coupling between the detector and the  field is weak, i.e., that $c$ is small, we are allowed to employ the Born approximation, which assumes that the state
  of the total system at time $\grt$ approximates a tensor product
\be
\hat \grr_{\text{tot}}(\grt)\approx \hat \grr(\grt)\otimes \hat \grr_{\hat{\Phi},\beta},
\ee
where $\hat \grr$ is the reduced density matrix of the detector, and we assume that the field is found in a thermal state $\hat{\grr}_{\hat{\Phi},\beta} = |\Omega_\beta \rangle \langle \Omega_\beta |$. 

We are interested in the evolution of the detector. Inserting the tensor product in Eq. (\ref{von}) and taking the partial trace over the degrees of freedom of the field, we obtain to leading order in the coupling constant an integrodifferential equation for the reduced density matrix of the detector
\begin{align}\label{mastereq}
\dot{\hat{\rho}}(\grt)& =c^2\int_{0}^{\grt}d\grt'\Big\{\Big(\hat{\grs}_-\hat{\grr}(\grt')\hat{\grs}_-e^{-2\ii\grw \grt}+\hat{\grs}_-\hat{\grr}(\grt')\hat{\grs}_+
-\hat{\grs}_+\hat{\grs}_-\hat{\grr}(\grt')\Big)e^{\ii\grw(\grt-\grt')}\nn\\&+
\Big(\hat{\grs}_+\hat{\grr}(\grt')\hat{\grs}_+e^{2\ii\grw \grt}+\hat{\grs}_+\hat{\grr}(\grt')\hat{\grs}_-
-\hat{\grs}_-\hat{\grs}_+\hat{\grr}(\grt')\Big)e^{-\ii\grw(\grt-\grt')}\Big\}\chi(\grt)\chi(\grt')\mathcal{W}^+_\beta(\grt -\grt')\nn\\&+ \text{H.c.},
\end{align}
where $\mathcal{W}^+_\beta(\grt -\grt')$ is the ``positive-frequency", or rather ``positive mixture", KMS Wightman function introduced above.

We next employ the Markov approximation, which assumes that the evolution of the state of the detector at some time $\tau'>\tau$ is completely determined by its present state $\rho(\tau)$. The environment (the field) does not retain any information flowed from the system (the detector) due to their interaction, i.e., there are no memory effects, and thus no information is transferred back to the system affecting its subsequent evolution. The Markov approximation is obtained if in Eq. (\ref{mastereq}) we replace the density matrix $\hat \rho(\grt')$ by $\hat \rho(\grt)$ and extend the upper limit of integration to infinity. 

 A word on the Markov approximation is due. In \cite{DMCA} the role played by non-Markovian effects in the Unruh effect has been studied. In particular, in Sec. IV C in that work it is pointed that, for a detector interacting constantly with the field, at early times, $\omega \tau = O(1)$, non-Markovian effects are important. In this sense, the Markovian approximation that we employ is well motivated, since we are concerned with late-time asymptotic states, in the same way that the standard derivation of the Unruh effect (see e.g. \cite{Birrell}) is justified, where the transition probability of the detector is calculated (within perturbation theory) for long interaction times.

Expressing the density operator in a matrix form we obtain for its elements
\begin{subequations}
\label{mrij}
\begin{align}\label{mr11}
 \dot{\grr}_{11}(\grt) =&-c^2 \grr_{11}(\grt)\int_0^{\infty}d\grt'\chi(\grt)\chi(\grt')\left[e^{\ii\grw (\grt-\grt') }\mathcal{W}_{\beta}^{+}(\grt-\grt')+e^{-\ii\grw (\grt-\grt')}\mathcal{W}_{\beta}^{-}(\grt-\grt')\right]
\nonumber \\
&+ c^2 \grr_{00}(\grt)\int_0^{\infty}d\grt'\chi(\grt)\chi(\grt')\left[e^{-\ii\grw (\grt-\grt')}\mathcal{W}_{\beta}^+(\grt-\grt')+e^{\ii\grw (\grt-\grt')}\mathcal{W}_{\beta}^-(\grt-\grt')\right],
\end{align}
\begin{align}\label{mr00}
\dot{\grr}_{00}(\grt) =&-c^2 \grr_{00}(\grt)\int_0^{\infty}d\grt'\chi(\grt)\chi(\grt')\left[e^{-\ii\grw (\grt-\grt') }\mathcal{W}_{\beta}^{+}(\grt-\grt')+e^{\ii\grw (\grt-\grt')}\mathcal{W}_{\beta}^{-}(\grt-\grt')\right]
\nonumber \\
&+ c^2 \grr_{11}(\grt)\int_0^{\infty}d\grt'\chi(\grt)\chi(\grt')\left[e^{\ii\grw (\grt-\grt')}\mathcal{W}_{\beta}^{+}(\grt-\grt')+e^{-\ii\grw (\grt-\grt')}\mathcal{W}_{\beta}^{-}(\grt-\grt')\right],
\end{align}
\begin{align}\label{mr01}
\dot{\grr}_{01}(\grt) =&-c^2 \grr_{01}(\grt)\int_0^{\infty}d\grt'\chi(\grt)\chi(\grt')\left[e^{-\ii\grw (\grt-\grt') }\mathcal{W}_{\beta}^{+}(\grt-\grt')+e^{-\ii\grw (\grt-\grt')}\mathcal{W}_{\beta}^{-}(\grt-\grt')\right]
\nonumber \\
&+ c^2 \grr_{10}(\grt)e^{-2\ii\grw\grt}\int_0^{\infty}d\grt'\chi(\grt)\chi(\grt')\left[e^{\ii\grw (\grt-\grt')}\mathcal{W}_{\beta}^{+}(\grt-\grt')+e^{\ii\grw (\grt-\grt')}\mathcal{W}_{\beta}^{-}(\grt-\grt')\right].
\end{align}
\end{subequations}

Equations \eqref{mr11}-\eqref{mr01} describe the time evolution of the reduced density matrix of the detector in the Born-Markov approximation for a general form of the switching function. 

Notice that from a mathematical standpoint the switching function, $\chi \in C_0^\infty(\mathbb{R})$ (multiplied by the exponential factors) plays the role of a test function for the distributions $\mathcal{W}_{\beta}^{\pm}$, and in this sense the Eq. \eqref{mrij} are unambiguously defined.  From a physical standpoint, it is natural to assume that the interaction between the detector and the field is switched on at proper time $\tau = 0$, so that ${\rm supp}(\chi) \subset \mathbb{R}^+_0$.


\section{Asymptotic equilibrium states and adiabatic scaling}
\label{Sec:Equilibrium}

We are interested in considering the emergence of thermality in the detector's reduced density matrix, $\hat \rho$, for long interaction times with the field in a KMS state. In this sense, and following \cite{BL}, we introduce a one-parameter family of adiabatically scaling switching functions. For $\lambda \geq 1$, we set
\be
\chi_{\grl}(\grt)=\chi(\grt/\grl),
\ee
where $\grl$ is a positive parameter that specifies the duration of the interaction, with the long-time asymptotic regime approached as $\lambda$ becomes large. The role of the parameter $\lambda$ is to stretch uniformly the duration of the interaction, including the switching times, during which the interaction between the field and the detector may vary in strength. This also means that the switching on and off are done slowly. 

We should mention that other choices of scaling for the switching are possible, such as, for example, the so-called plateau scaling in \cite{BL}. In this case, the switching scales leaving the tails fixed, i.e., the interaction is switched on, the system interacts constantly for a long time, and then switch off. In \cite{BL}, it was shown that for the Unruh effect, the thermalisation of the response function cannot be polynomial in time at large frequency gap if a plateau scaling is chosen. For this reason, the adiabatic scaling that we discuss here is more interesting. 

Corresponding with the scaling of the switching functions, we obtain a one-parameter family of density operators, which we denote, $\hat \rho^{(\lambda)}$ , which in the Born-Markov approximation satisfies the dynamical equations
\begin{subequations}
\label{mrij-Lambda}
\begin{align}\label{mr11-Lambda}
\dot{\grr}^{(\lambda)}_{11}(\grt) =&-c^2 \grr^{(\lambda)}_{11}( \grt)\int_0^{\infty}d\grt'\chi_\lambda(\grt)\chi_\lambda(\grt')\left[e^{\ii\grw (\grt-\grt') }\mathcal{W}_{\beta}^{+}(\grt-\grt')+e^{-\ii\grw (\grt-\grt')}\mathcal{W}_{\beta}^{-}(\grt-\grt')\right]
\nonumber \\
&+ c^2 \grr^{(\lambda)}_{00}( \grt)\int_0^{\infty}d\grt'\chi_\lambda(\grt)\chi_\lambda(\grt')\left[e^{-\ii\grw (\grt-\grt')}\mathcal{W}_{\beta}^+(\grt-\grt')+e^{\ii\grw (\grt-\grt')}\mathcal{W}_{\beta}^-(\grt-\grt')\right],
\end{align}
\begin{align}\label{mr00-Lambda}
\dot{\grr}^{(\lambda)}_{00}(\grt) =&-c^2 \grr^{(\lambda)}_{00}( \grt)\int_0^{\infty}d\grt'\chi_\lambda(\grt)\chi_\lambda(\grt')\left[e^{-\ii\grw (\grt-\grt') }\mathcal{W}_{\beta}^{+}(\grt-\grt')+e^{\ii\grw (\grt-\grt')}\mathcal{W}_{\beta}^{-}(\grt-\grt')\right]
\nonumber \\
&+ c^2 \grr^{(\lambda)}_{11}( \grt)\int_0^{\infty}d\grt'\chi_\lambda( \grt)\chi_\lambda(\grt')\left[e^{\ii\grw (\grt-\grt')}\mathcal{W}_{\beta}^{+}(\grt-\grt')+e^{-\ii\grw (\grt-\grt')}\mathcal{W}_{\beta}^{-}(\grt-\grt')\right],
\end{align}
\begin{align}\label{mr01-Lambda}
\dot{\grr}^{(\lambda)}_{01}( \grt) =&-c^2 \grr^{(\lambda)}_{01}( \grt)\int_0^{\infty}d\grt'\chi_\lambda(\grt)\chi_\lambda(\grt')\left[e^{-\ii\grw (\grt-\grt') }\mathcal{W}_{\beta}^{+}(\grt-\grt')+e^{-\ii\grw (\grt-\grt')}\mathcal{W}_{\beta}^{-}(\grt-\grt')\right]
\nonumber \\
&+ c^2 \grr^{(\lambda)}_{10}( \grt)e^{-2\ii\grw\grt}\int_0^{\infty}d\grt'\chi_\lambda(\grt)\chi_\lambda(\grt')\left[e^{\ii\grw (\grt-\grt')}\mathcal{W}_{\beta}^{+}(\grt-\grt')+e^{\ii\grw (\grt-\grt')}\mathcal{W}_{\beta}^{-}(\grt-\grt')\right].
\end{align}
\end{subequations}

\subsection{Asymptotic limit, $\lambda \to \infty$}

Since we are interested in equilibrium states as the interaction duration becomes long, we study the limits $\hat{\rho}^{(\infty)} = \lim_{\lambda \to \infty} \hat{\rho}^{(\lambda)}$ in Eq.\eqref{mrij-Lambda}. This can be done as follows. Consider the integral
\be
\Lambda^\pm_{\grl}(\pm \omega, \grt)= \int_{-\infty}^{\infty}\dd\grt'\chi_{\grl}(\grt)\chi_{\grl}(\grt')e^{\pm \ii \grw (\grt-\grt') }\mathcal{W}_{\beta}^{\pm}(\grt-\grt'),
\ee
which appears in Eq. \eqref{mrij-Lambda}, where we have extended the integration limits due to the support properties of $\chi$,  where we have assumed that ${\rm supp}(\chi) \subset \mathbb{R}^+_0$, and hence that $\chi$ vanishes for negative argument. We know that the Fourier transform of $\chi$ exists and decays faster than any polynomial due to the Cauchy-Paley-Wiener theorem, hence we can write
\begin{align}
\Lambda^\pm_{\grl}(\pm \omega, \grt) & = \frac{1}{(2 \pi)^2} \int_{-\infty}^\infty \! \dd \omega' \int_{-\infty}^\infty \! \dd \omega'' \widehat{\chi}(\omega') \widehat{\chi}(-\omega'') \int_{-\infty}^{\infty} \! \dd\grt'  \expe^{\ii \omega' \tau/\lambda - \ii \omega'' \tau'/\lambda} \expe^{\pm \ii \grw (\grt-\grt') }\mathcal{W}_{\beta}^{\pm}(\grt-\grt') \nonumber \\
& = \frac{1}{(2 \pi)^2} \int_{-\infty}^\infty \! \dd \omega' \int_{-\infty}^\infty \! \dd \omega'' \widehat{\chi}(\omega') \widehat{\chi}(-\omega'') \int_{-\infty}^{\infty} \! \dd s \expe^{\ii(\omega' - \omega'') \tau/\lambda} \expe^{\ii \omega'' s/\lambda} \expe^{\pm \ii \omega s} \mathcal{W}_{\beta}^{\pm}(s),
\end{align}
where we have performed the change of variables $s = \tau - \tau'$, and denoted the Fourier transform by a wide caret, $\widehat{\chi} = \mathcal{F}[\chi]$, cf. Sec. \ref{Notation}. We assume that the Fourier transform of the pulled-back ``positive" and ``negative frequency" Wightman functions, $\widehat{\mathcal{W}}_\beta^\pm$, exist and satisfy the requirements of Prop. \ref{PropDBC} (see also the remarks below Prop. \ref{PropDBC}) and further write
\begin{align}
\Lambda^\pm_{\grl}(\pm \omega, \grt) &  = \int_{-\infty}^\infty \!\! \frac{\dd \omega'}{(2 \pi)^3} \int_{-\infty}^\infty \!\!\!\! \dd \omega'' \int_{-\infty}^\infty \!\!\!\! \dd \omega''' \widehat{\chi}(\omega') \widehat{\chi}(-\omega'') \widehat{\mathcal{W}}_{\beta}^{\pm}(\omega''') \int_{-\infty}^{\infty} \!\!\!\! \dd s \expe^{\ii(\omega' - \omega'') \tau/\lambda} \expe^{\ii \left( \frac{\omega''}{\lambda} \pm \omega + \omega'''\right)  s} ,
\end{align}
which yields, using the integral representation of the delta function,
\begin{align}
\Lambda^\pm_{\grl}(\pm \omega, \grt) &  = \frac{1}{(2 \pi)^2} \int_{-\infty}^\infty \!\!\! \dd \omega' \int_{-\infty}^\infty \!\!\! \dd \omega'' \widehat{\chi}(\omega') \widehat{\chi}(-\omega'') \widehat{\mathcal{W}}_{\beta}^{\pm}(\mp \omega  - \omega''/\lambda)  \expe^{\ii(\omega' - \omega'') \tau/\lambda}.
\label{Lambdapm}
\end{align}

We now seek to compute $\Lambda^\pm_{\infty}(\omega, \grt) = \lim_{\lambda \to \infty} \Lambda^\pm_{\grl}(\omega, \grt)$ by taking the limit inside the integrals on the right-hand side of Eq. \eqref{Lambdapm}. To this end, define
\be
F_{\grl}(\pm  \grw, \tau)=\int_{-\infty}^{\infty}d\grw''\widehat{\chi}(-\grw'')e^{-\ii \grw''\grt/\grl}\widehat{\mathcal{W}}_{\beta}^\pm\left(\mp \grw-\frac{\grw''}{\grl}\right),
\label{Fdef}
\ee
whereby
\be\label{Lint}
\Lambda^\pm_{\grl}(\pm \omega, \grt)= \frac{1}{(2 \pi)^2}\int_{-\infty}^{\infty}d\grw'\widehat{\chi}(\grw')e^{\ii\grw'\grt/\grl}F_{\grl}(\pm \grw, \tau).
\ee

The strategy is to apply the dominated convergence theorem. This step first entails finding an integrable, $\lambda$-independent bound for $F_\lambda(\pm \omega,\tau)$. Indeed, 
\be
|F_{\grl}(\pm \grw, \tau)|\leq \int_{-\infty}^{\infty}d\grw''|\widehat{\chi}(-\grw'')|\left(A+B\left(|\grw|+|\grw''|\right)\right)^n = G(\omega),
\label{Fbound}
\ee
where $A,B\in\mathbb{R}$ are positive and $n\in\mathbb{Z}$, by the polynomial bound of $\widehat{\mathcal{W}}_\beta^\pm$. The integrand is integrable due to the rapid decay of $\widehat{\chi}$ and yields a $\lambda$-independent and $\omega'$-independent function of $\omega$, which we called $G$. Since $\widehat{\chi}(\grw')e^{\ii\grw'\grt/\grl} \leq |\widehat{\chi}(\grw')|$, we have that the integrand on the right-hand side of Eq. \eqref{Lint} is bounded by a $\lambda$-independent integrable function, $G(\omega)|\widehat{\chi}(\grw')|$, from where it follows that
\be
\Lambda_{\infty}^{\pm}(\pm \omega, \grt)= \frac{1}{(2 \pi)^2}\int_{-\infty}^{\infty}d\grw' \widehat{\chi}(\grw') \lim_{\lambda \to \infty} F_{ \grl}(\pm  \grw, \tau),
\ee
but we have seen that the integrand of the right-hand side of Eq. \eqref{Fdef} is bounded by a $\lambda$-independent integrable function cf. the middle expression of \eqref{Fbound}, so we obtain that the limit can be taken inside the integrals, and
\begin{equation}
\Lambda^\pm_{\infty}(\pm \omega, \grt) = \frac{1}{(2 \pi)^2} \widehat{\mathcal{W}}_{\beta}^{\pm}(\mp \omega) || \widehat{\chi} ||^2_{L^1},
\end{equation}
where the $L^1$-norm over the reals is defined in the usual way as $|| \widehat{\chi} ||_{L^1} = \int_{-\infty}^{\infty}d\grw' \widehat{\chi}(\omega')$.

Equation \eqref{mrij-Lambda} then reads as $\lambda \to \infty$ as
\begin{subequations}
\label{mrij-infty}
\begin{align}
\dot{\grr}^{(\infty)}_{11}(\grt) =& \frac{c^2 || \widehat{\chi} ||^2_{L^1}}{(2\pi)^2}  \left[ - \grr^{(\infty)}_{11} ( \grt) \left( \widehat{\mathcal{W}}_{\beta}^{+}(- \omega) + \widehat{\mathcal{W}}_{\beta}^{-}(\omega) \right)
+ \grr^{(\infty)}_{00}(\tau) \left(\widehat{\mathcal{W}}_{\beta}^{+}(\omega) + \widehat{\mathcal{W}}_{\beta}^{-}(- \omega) \right) 
\right], \label{mr11-infty} \\
\dot{\grr}^{(\infty)}_{00}(\grt) =&\frac{c^2 || \widehat{\chi} ||^2_{L^1}}{(2\pi)^2}  \left[ + \grr^{(\infty)}_{11} ( \grt) \left( \widehat{\mathcal{W}}_{\beta}^{+}(- \omega) + \widehat{\mathcal{W}}_{\beta}^{-}(\omega) \right)
- \grr^{(\infty)}_{00}(\tau) \left(\widehat{\mathcal{W}}_{\beta}^{+}(\omega) + \widehat{\mathcal{W}}_{\beta}^{-}(- \omega) \right) 
\right], \label{mr00-infty}  \\
\dot{\grr}^{(\infty)}_{01}( \grt) = &\frac{c^2 || \widehat{\chi} ||^2_{L^1}}{(2\pi)^2} 
\left[- \grr^{(\infty)}_{01} (\tau) \left(\widehat{\mathcal{W}}_{\beta}^{+}(\omega) + \widehat{\mathcal{W}}_{\beta}^{-}(\omega) \right) +  \grr^{(\infty)}_{10} (\tau)\expe^{-2 \ii \omega \tau} \left(\widehat{\mathcal{W}}_{\beta}^{+}(- \omega) + \widehat{\mathcal{W}}_{\beta}^{-}(- \omega) \right) \right]. \label{mr01-infty}
\end{align}
\end{subequations}

\subsection{Stationary states as $\lambda \to \infty$}
\label{sec:stat}

According to Appendix \ref{app:Stat}, in the long interaction time limit the detector relaxes to a steady state,  
\be
\lim_{\lambda\to\infty}\frac{d\hat{\rho}^{(\lambda)} (\tau)}{d\tau}=0,
\ee
which reduces the system of Eq. \eqref{mrij-infty} to algebraic relations that can be solved for the components of $\hat \rho^{(\infty)}$. Using Eq. \eqref{mr01-infty} and its complex conjugate, it can be seen that the off-diagonal coefficients of $\hat \rho^{(\infty)}$ must vanish, in agreement with Appendix \ref{app:Stat}, while Eq. \eqref{mr11-infty} and \eqref{mr00-infty} yield a relation between $\hat \rho_{00}^{(\infty)}$ and $\hat \rho_{11}^{(\infty)}$ as follows,
\begin{subequations}
\begin{align} 
\rho_{01}^{(\infty)}& =0, \label{Solve01}\\
\grr_{11}^{(\infty)}\left(\widehat{\mathcal{W}}^+(-\grw)+\widehat{\mathcal{W}}^-(\grw)\right)-
\grr_{00}^{(\infty)}\left(\widehat{\mathcal{W}}^+(\grw)+\widehat{\mathcal{W}}^-(-\grw)\right) & = 0 \label{Solve0011}
\end{align}
\end{subequations}

Note that Eq. \eqref{Solve01} is expected; the off-diagonal elements of the density matrix vanish since interaction with the field environment leads to the decoherence of any initial state of the detector. In order to solve Eq. \eqref{Solve0011} we use the following lemma:

\begin{lemma}
\label{Lem}
Let $\mathcal{W}_\mathbb{C}$ be the complex-time two-point function associated to the worldline-pullback two-point function $\mathcal{W}^+_\beta$ at inverse temperature $\beta$, such that $\mathcal{W}_\mathbb{C}$ and $P(\omega)$ are as in Prop. \ref{PropDBC} and $\widehat{\mathcal{W}}_\mathbb{C} \leq P(\omega)$ its Fourier transform satisfying the detailed balance condition \eqref{DBC}. If $\widetilde{\mathcal{W}}_\mathbb{C}$ is the complex-time two-point function associated to the worldline-pullback two-point function $\mathcal{W}^-_\beta$, then its Fourier transform, $\widehat{\widetilde{\mathcal{W}}}_\mathbb{C}$, satisfies the following ``conjugate" detailed balance condition
\begin{equation}
\widehat{\widetilde{\mathcal{W}}}_\mathbb{C}(\omega) = \expe^{\beta \omega} \widehat{\widetilde{\mathcal{W}}}_\mathbb{C}(-\omega)
\end{equation}
\end{lemma}
\begin{proof}
The proof follows from the relation $\mathcal{W}^+_\beta = \overline{\mathcal{W}^-_\beta}$ and Prop. \ref{PropDBC}.
\end{proof}

From Lemma \ref{Lem}, Eq. \eqref{Solve0011} becomes
\begin{align}
\left( \grr_{11}^{(\infty)} -
\expe^{-\beta \omega} \grr_{00}^{(\infty)} \right)\left(\widehat{\mathcal{W}}^+(-\grw)+\widehat{\mathcal{W}}^-(\grw)\right)= 0,
\end{align}
which implies that $\grr_{00}^{(\infty)} = \expe^{\beta \omega} \grr_{11}^{(\infty)}$. Introducing the normalisation $\grr_{00}^{(\infty)} + \grr_{11}^{(\infty)} = 1$, we have that
\begin{align}
    \hat \rho^{(\infty)}
       = \frac{1}{1 + \expe^{-\beta \omega}} \begin{pmatrix} 
       1 & 0\\ 0 & \expe^{-\beta \omega} \\
       \end{pmatrix}       
\label{rhoinfty}
\end{align}
from where we conclude that the density operator elements of the detector state are in a Gibbs equilibrium state at inverse temperature $\beta$,
\be
\hat{\rho}^{(\infty)}=\frac{e^{-\beta \hat{H}_{\rm D}}}{{\rm tr} \left( e^{-\beta \hat{H}_{\rm D}} \right)}.
\ee

 Thus, a detector interacting with a quantum field that is in a thermal KMS state for a long time will eventually reach an equilibrium state at temperature $T = 1/\beta$, independently of the details of the switching of the interaction.

\subsection{Frequency-dependent effective temperature}

Throughout the discussion of thermal states in Sec. \ref{sec:KMS} we have assumed that the inverse temperature, $\beta$, is constant. Our arguments leading to stationarity, however, do not rely at any point on $\beta$ being constant. It is straightforward, therefore, to allow $\beta$ to be a function $\omega \mapsto \beta(\omega)$. While some might be led to argue that the equivalence between the detailed balance condition and the KMS condition is not evident in this case, we wish to stress that an operational meaning can be attached to such frequency-dependent function $\beta(\omega)$ in the sense that a stationary detector does reach equilibrium at late times, regardless of whether $\beta$ is constant in the frequency spectrum. To such achieved equilibrium, one can associate an effective temperature, such that the following detailed balance conditions hold
\begin{equation}
\widehat{\mathcal{W}}_{\beta}^{+}(-\omega) = \expe^{\beta(\omega) \omega} \widehat{\mathcal{W}}_{\beta}^{+}(\omega) \text{ \, and \, } \widehat{\mathcal{W}}_{\beta}^{-}(\omega) = \expe^{\beta(\omega) \omega} \widehat{\mathcal{W}}_{\beta}^{-}(-\omega).
\label{DBCeff}
\end{equation}

 Conditions \eqref{DBCeff} are rather mild, for provided that the ratio of $\widehat{\mathcal{W}}(\omega)$ and $\widehat{\mathcal{W}}(-\omega)$ be positive, one can define the frequency-dependent temperature by solving for $\beta(\omega)$ in any of the relations appearing in \eqref{DBCeff}. But that the aforementioned ratio is positive is not a stringent requirement by the positivity of the state and the Bochner-Schwarz theorem, as we have discussed in the final paragraph of Sec. \ref{sec:KMS}. That the two relations appearing in \eqref{DBCeff} hold simultaneously follows from applying Lemma \ref{Lem} pointwise in $\omega$.

Indeed, a constant temperature will be achieved if (but not only if) the stationary trajectory of the detector is along the worldlines generated by the Killing vector field that also defines the positive and negative frequencies of the quantum field theory, but this need not be the case for every stationary trajectory that the detector can follow. 

For example, in Minkowski spacetime when the field is in the Minkowski vacuum, i.e., with a notion of positive frequency with respect to timelike translations, in addition to (i) inertial motion (with timelike translations as associated Killing vectors), one has the following stationary trajectories that do not coincide with the notion of positive frequency of the state: (ii) linear uniform acceleration (with boosts as associated Killing vectors), (iii) circular motion (linear combination of translation and rotation) , (iv) cusped motion (translation and null rotation), (v) motion generated by a linear combination of boosts with spatial translation and (vi) motion generated by a linear combination of boosts with rotation \cite{Letaw} (see also \cite{Louko:2006zv}).

Of the six stationary trajectories described above, it is known that detectors following trajectories (i) and (ii) will respond at a constant temperature, resp. (i) zero (in a limiting sense) and (ii) positive (the Unruh effect). In cases (iii) and (iv), there are arguments that a detector will thermalise at a frequency-dependent effective temperature \cite{Bell:Leinaas, Unruh:circ}. By our general arguments in Sec. \ref{sec:stat}, such effective temperatures can also be associated to detectors along trajectories (v) and (vi), provided that Eq. \eqref{DBCeff} hold.


\section{Applications}\label{apps}

In this section, we give a few examples of detector spacetime trajectories that lead to the thermalisation (at constant or frequency-dependent temperature) of the detector interacting with a scalar field in a fixed background with an interaction Hamiltonian of the form \eqref{Hint} with $\chi$ replaced by $\chi_\lambda$, and in the long-time limit $\lambda \to \infty$. From our results in Sec. \ref{Sec:Equilibrium}, it suffices that the pullback of the two-point function satisfy the (generalised) detailed balance condition for the detector to reach a thermal equilibrium state as $\lambda \to \infty$.

\subsection{The Unruh and Hawking effects}

\subsubsection*{The Unruh effect}

It is well-known that the Minkowski vacuum restricts to a thermal state on the right (or left) Rindler wedge. Consider a linearly, uniformly accelerated detector with trajectory parametrised by its proper time $\tau$ as
\begin{align}
{\rm x}(\grt)=\left(a^{-1}\sinh(a\grt), \, a^{-1}\cosh(a\grt),\, 0, \, 0\right),
\label{Accel}
\end{align}
where $a$ is the proper acceleration. The detector interacts with a massless scalar Klein-Gordon field. Then, the pullback of the Minkowski Wightman functions to the detector trajectory \eqref{Accel} are given by
\be\label{Wacc}
\mathcal{W}^{\pm}(\grt -\grt')= -\lim_{\epsilon \rightarrow 0^+} \frac{a^2}{16\pi^2\sinh^2[a(\grt-\grt' \mp i\gre)/2]}.
\ee
Their Fourier transforms of $\mathcal{W}^{\pm}$ are given by
\be
\widehat{\mathcal{W}}^+(\grw)=\frac{\grw}{2\pi}\frac{1}{e^{\frac{2\pi\grw}{a}}-1}, \quad \widehat{\mathcal{W}}^-(\grw)=\frac{\grw}{2\pi}\frac{e^{\frac{2\pi\grw}{a}}}{e^{\frac{2\pi\grw}{a}}-1}.
\ee
and satisfy the detailed balance condition at the Unruh temperature, $T_{\rm U} = a/(2\pi)$. It follows from the results of Sec. \ref{Sec:Equilibrium}, and Eq. \eqref{rhoinfty} in particular, that as the interaction between the detector and the field becomes long, $\lambda \to \infty$, we have that 
\be
\hat{\rho}^{(\infty)} = \frac{1}{1 + \expe^{-2 \pi \omega/a} }\left( \begin{array}{cc} 1 & 0 \\0 &e^{-\frac{2\pi\grw}{a}} \end{array} \right), \label{asympt}
\ee
which is a Gibbs  state at temperature $T_{\rm U} =a/(2\pi)$. This is the Unruh effect: A uniformly accelerated detector perceives the field vacuum state as a thermal state. As a result, the detector relaxes to an equilibrium Gibbs state at the Unruh temperature, which is an asymptotic solution of the master equation for the reduced density matrix of the detector. See e.g. \cite{DM,DMCA,BL,DeBievre,Benatti}.

\subsubsection*{The Hawking effect}

Consider now an eternal Schwarzschild black hole, which is an idealised situation for a primordial black hole with zero angular momentum, and a detector following a constant radial trajectory, say $r = R > 2M$, interacting with a Klein-Gordon field in the Hartle-Hawking-Israel (HHI) state.

It is well known that the HHI state restricts to a thermal state in the exterior region (region I) of the Kruskal maximal analytic extension of Schwarzschild spacetime \cite{Kay:1988mu, Sanders:2013vza}, where the metric tensor takes the form $g = -(1-2M/r) dt^2 + (1-2M/r)^{-1} dr + r^2 d\Omega^2$, from where it is clear that the exterior region is static with Killing vector field $\xi = \partial_t$. The radially constant trajectory $r = R > 2M$ is an orbit of the Killing vector field $\xi$ in the exterior region, and is hence a stationary orbit. Without explicitly knowing the precise details of the two-point function, we nevertheless know from these considerations that along the detector's worldline, the following detailed balance condition holds
\begin{equation}
\widehat{\mathcal{W}}^+(-\grw) = \expe^{\omega/ T_{\rm loc}} \widehat{\mathcal{W}}^+(\grw), \hspace{1cm} T_{\rm loc} =  \frac{1}{8 \pi M (1-2M/r)^{1/2}}.
\end{equation}

If $\widehat{\mathcal{W}}(\omega) \leq P(\omega)$, where $P(\omega)$ is a polynomial function (and the Fourier transform is not simply a polynomially bounded positive {\it measure}), it follows from our results in Sec. \ref{Sec:Equilibrium}, cf. Eq. \eqref{rhoinfty}, that as $\lambda \to \infty$,
\begin{align}
    \hat \rho^{(\infty)}
       = \frac{1}{1 + \expe^{-\omega/T_{\rm loc}}} \begin{pmatrix} 
       1 & 0\\ 0 & \expe^{-\omega/T_{\rm loc}} \\
       \end{pmatrix}.     
\end{align}

\subsection{Cusped motion}

Consider a detector interacting with a massless Klein-Gordon field in the Minkowski vacuum that follows the stationary, cusped trajectory \cite{Letaw}
\begin{align}
x^{\mu}(\tau)=\left(\tau+\frac{1}{6}a^2\tau^3,\, \frac{1}{2}a\tau^2,\, \frac{1}{6}a^2\tau^3,0\right).
\label{cusp}
\end{align}

The pullback of the Wightman functions to the trajectory are
\begin{align}
\mathcal{W}^{\pm}_{\rm cusp}(\tau-\tau')=-\frac{1}{4\pi^2}\lim_{\epsilon\to 0}\frac{1}{(\tau-\tau'\mp i\epsilon)^2+\frac{a^2}{12}(\tau-\tau')^4},
\end{align}
which can be written as 
\begin{align}\label{Wcusp}
\mathcal{W}^{\pm}_{\rm cusp}(\tau-\tau')=-\frac{1}{4\pi^2}\lim_{\epsilon\to 0}\left\{\frac{1}{(\tau-\tau'\mp i\epsilon)^2}-\frac{a^2}{12+a^2(\tau-\tau')^2}\right\},
\end{align}
from where their Fourier transform can be evaluated using standard complex-analytic techniques, yielding
\begin{equation}
\widehat{\mathcal{W}}^{\pm}_{\rm cusp}(\omega) = -\frac{\omega}{2 \pi} \Theta(\mp \omega) + \frac{a}{4 (12)^{1/2} \pi} \exp\left( - (12)^{1/2} |\omega|/a \right).
\label{CuspFourier}
\end{equation}

We should stress that some results reported in the literature (e.g. \cite{Letaw, Good:2017xzl, Rad:Singl}) for a detector following the cusped trajectory \eqref{cusp} would seem to suggest that the detector becomes thermalised at a temperature $a/(12)^{1/2}$. However, our result \eqref{CuspFourier}, which is in agreement with \cite[Eq. (5.13)]{Louko:2006zv}, suggests that as $\lambda \to \infty$, the density matrix for the detector becomes
\begin{align}
    \hat \rho^{(\infty)}
       = \frac{1}{1 + \expe^{-\omega/T_{\rm eff}(\omega)}} \begin{pmatrix} 
       1 & 0\\ 0 & \expe^{-\omega/T_{\rm eff}(\omega)} \\
       \end{pmatrix},     
\end{align}
where the frequency-dependent effective temperature, $T_{\rm eff}(\omega)$, is given by
\begin{equation}\label{Teff}
T^{-1}_{\rm eff}(\omega) = \beta(\omega) = \frac{1}{\omega} \log \left( \frac{\frac{\omega}{2 \pi} \Theta(\omega) + \frac{a}{4 (12)^{1/2} \pi} \exp\left( - (12)^{1/2} |\omega|/a \right)}{-\frac{\omega}{2 \pi} \Theta(-\omega) + \frac{a}{4 (12)^{1/2} \pi} \exp\left( - (12)^{1/2} |\omega|/a \right)} \right).
\end{equation}
Equation \eqref{Teff} is plotted in Fig. \ref{Teff:plot}.

\begin{figure}
\centering
\includegraphics[width=0.7\linewidth]{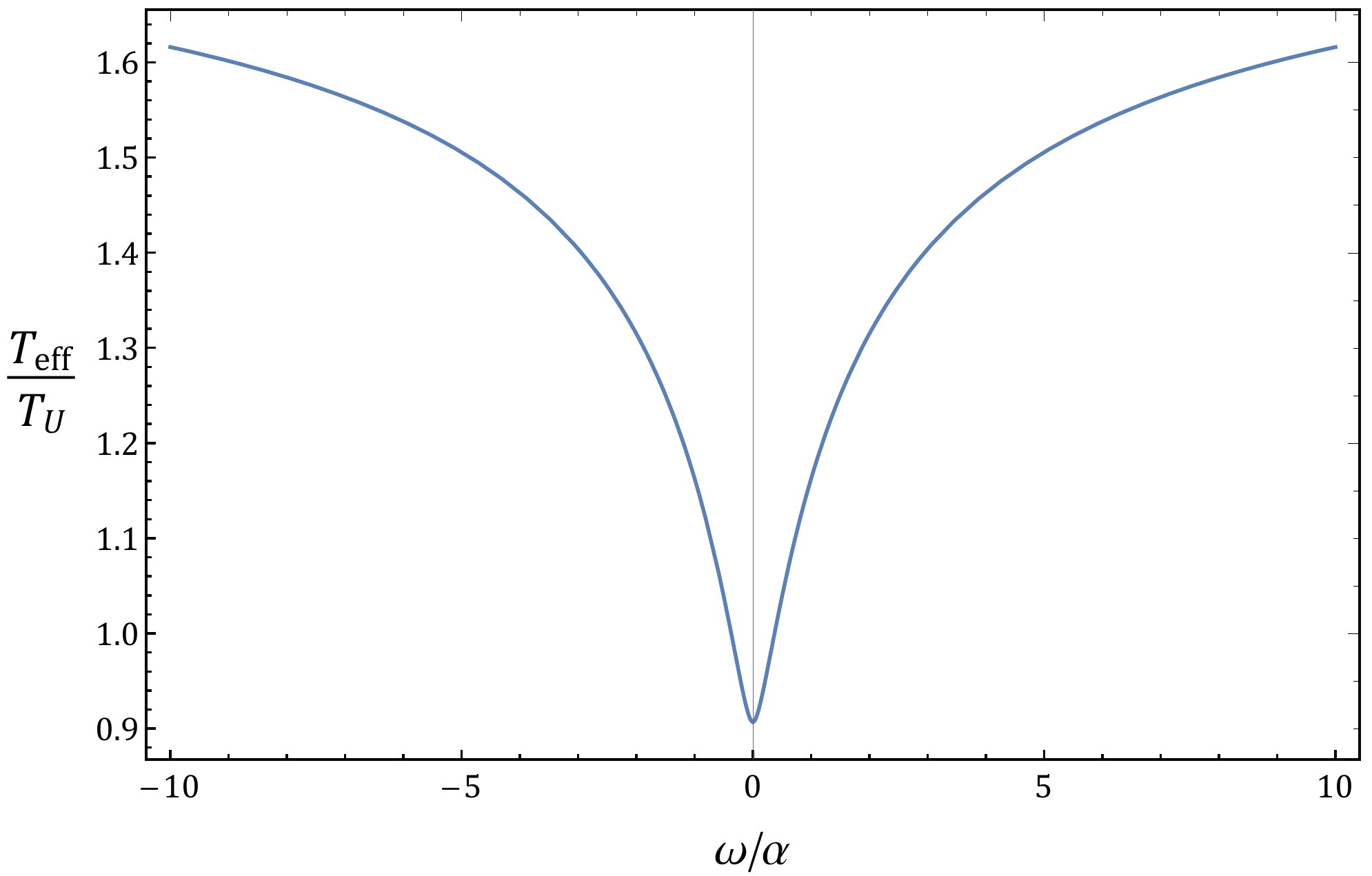}
\caption{The effective temperature \eqref{Teff} scaled to the Unruh temperature $T_U$ plotted against $\omega/a$ for fixed $a > 0$. Notice that for $\omega \ll a$ the Unruh temperature is higher, while for $\omega \gg a$ it is the effective temperature along the cusped motion that dominates.}\label{Teff:plot}
\end{figure}

\subsection{Circular motion}

We consider a detector following the circular trajectory \cite{LePf,Letaw}
\begin{align}
x^{\mu}(\tau)=\left(\gamma \tau,\, r\cos(\gamma \Omega \tau),\, r\sin(\gamma \Omega \tau) ,0\right),
\end{align}
in Minkowski spacetime, where  $\Omega$ is the angular velocity of the detector, $r$ is the radius of the circular orbit, $\gamma=1/\sqrt{1-\upsilon^2}$ is the Lorentz factor and $\upsilon=\Omega r$  is the velocity of the detector. 

The pullback of the Minkowski vacuum Wightman functions of a massless scalar field to the circular trajectory are 
\begin{align}\label{Wight:circ}
\mathcal{W}^{\pm}_{\rm circ}(\tau-\tau')=-\frac{1}{4\pi^2}\lim_{\epsilon\to 0}\frac{1}{\left(\gamma(\tau-\tau')\mp i\epsilon\right)^2-4r^2\sin^2\left[\gamma\Omega (\tau-\tau')/2\right]}.
\end{align}
As before, we seek to associate the effective temperature 
\begin{align}\label{T:circ}
T^{-1}_{\rm eff}(\omega) =  \frac{1}{\omega} \log \left(\frac{\widehat{\mathcal{W}}^+_{\rm circ}(-\omega)}{\widehat{\mathcal{W}}^+_{\rm circ}(\omega)}\right),
\end{align}
at which the circularly moving detector thermalises as $\lambda\to\infty$. While we do not evaluate the Fourier transform of the Wightman functions \eqref{Wight:circ} analytically, we present numerical results of the effective temperature \eqref{T:circ} in Fig \ref{Teff:circ}. 

We should mention that obtaining an analytic expression for this Fourier transform, and hence for the effective frequency dependent temperature in the circular motion case, amounts to summing the residues of the integrand defining the Fourier transform of $\mathcal{W}^{\pm}_{\rm circ}(z)$,  Eq. \eqref{Wight:circ}, as an integral over the complex $z$-plane. It can be seen that the denominator on the right-hand side of Eq. \eqref{Wight:circ} contains infinitely many roots on the complex plane, and hence the Fourier transform becomes an infinite sum over residues. Through a series of approximations, Unruh has argued that at large values of $\gamma$, i.e., in the limit of ultra-relativistic motion $\upsilon\to1$, and at large values of $\omega$, the detector will read off a temperature that is to leading order that of the cusped motion. See the preprint version of \cite{Unruh:circ}. A similar conclusion has been reached by Bell and Leinaas in \cite[Sec. 5]{Bell:Leinaas} through formal manipulations, where particularly relevant is the discussion around Eq. (57) in that article.

\begin{figure}
\centering
\includegraphics[width=0.8\linewidth]{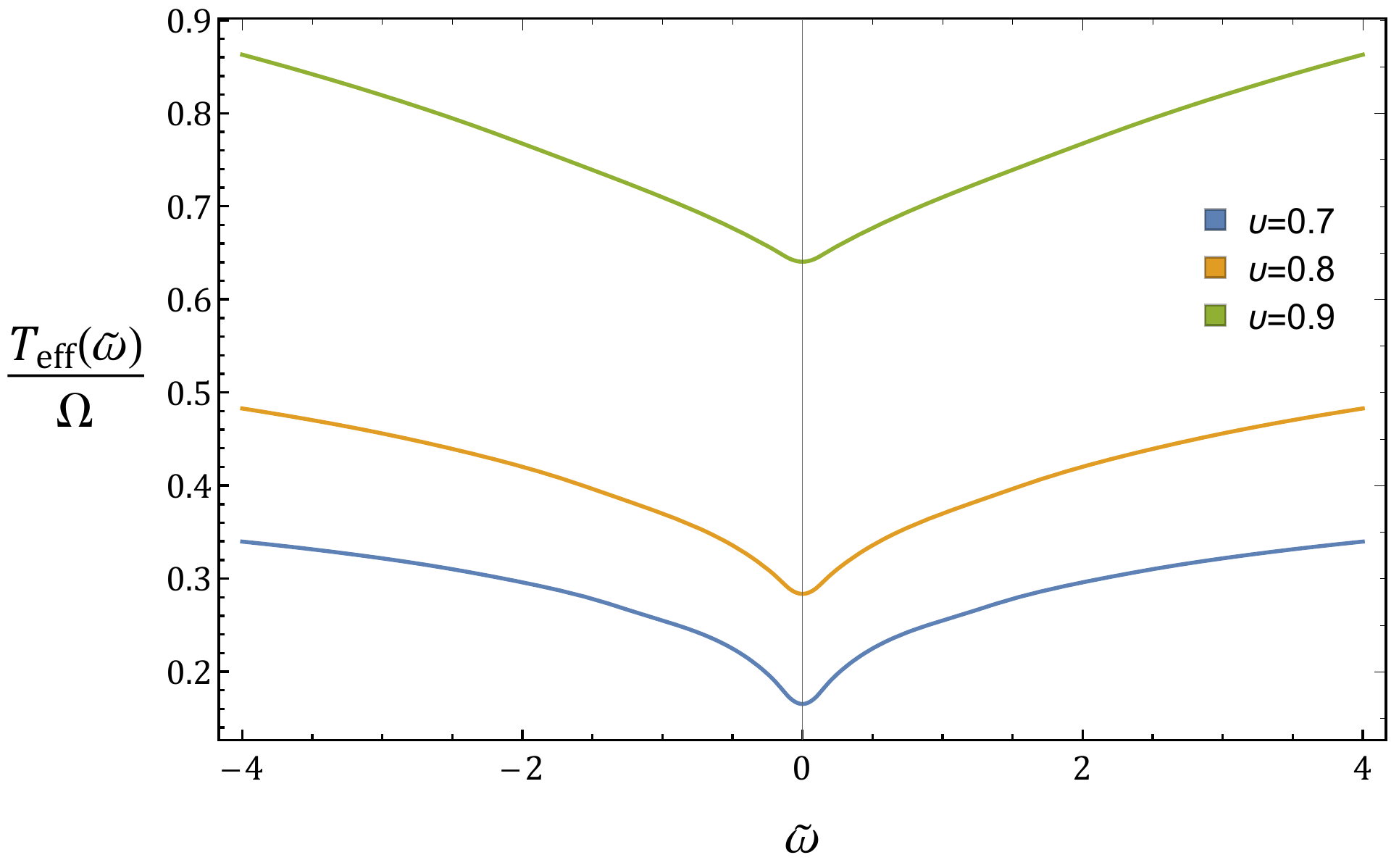}
\caption{The effective temperature \eqref{T:circ} scaled to the angular velocity $\Omega$ of the detector plotted against $\tilde{\omega}=\omega/\Omega$ for three different values of the velocity $\upsilon$.}\label{Teff:circ}
\end{figure}


\section{Conclusions}\label{concl}

We have first studied the relaxation to equilibrium of stationary detectors interacting with quantum fields in KMS states as open systems, and have shown that, in the Born-Markov approximation, detectors relax to thermal states at the KMS temperature. This has been done by dealing with a one-parameter family of switching functions that control the interaction between the detector and the field, and considering the long interaction time limit as the support of the switching function becomes stretched as the switching family parameter becomes large.

We have then relaxed the KMS condition for the quantum fields, requiring only that a frequency-dependent version of the detailed balance condition is satisfied along the detector's worldline, and concluded that, while the detectors do reach an equilibrium state in the asymptotic future, they become thermalised at a frequency-dependent, effective temperature. This situation seems to be relevant in describing the final state of stationary detectors that do not follow the orbit of the timelike Killing vector that defines the positive- and negative-frequency notions of the quantum field in a stationary spacetime. We have exemplified this situation by studying analytically the case of a detector following a cusped motion trajectory in Minkowski spacetime, as well as numerically for a rotating detector also in Minkowski spacetime, and in both cases interacting with a massless Klein-Gordon field in the Minkowski vacuum.

As a perspective, it should be interesting to go beyond the Born-Markov approximation, and to see whether a similar argument for the relaxation of the detectors to equilibrium (Gibbs or frequency-dependent thermal) states can be obtained. In the case of the Unruh effect it has been found in \cite{{DMCA}} that the late-time asymptotic state of the detector is a Gibbs thermal state even if non-Markovian effects are taken into account. This suggests that it is plausible that a similar conclusion can be reached for detectors interacting with fields in arbitrary KMS states.

\begin{acknowledgments}

B.A.J.-A. is supported by a Direcci\'on Genereral de
Asuntos del Personal Acad\'emico, Universidad Nacional
Autónoma de M\'exico (DGAPA-UNAM) postdoctoral fellowship held at Mexico City, and acknowledges additional support from the Sistema Nacional de
Investigadores, Consejo Nacional de Ciencia y
Tecnología (SNI-CONACYT), Mexico. D. M.'s research has been co-financed via a programme of State Scholarships Foundation (IKY) by the European Union (European Social Fund-ESF) and by Greek national funds through the action entitled ``Scholarships programme for postgraduates studies -2nd Study Cycle" in the framework of the Operational Programme ``Human Resources Development Program, Education and Lifelong Learning" of the National Strategic Reference Framework (NSRF) 2014--2020. D. M. acknowledges the Institute for Quantum Optics and Quantum Information - IQOQI Vienna and the Erwin Schr\"odinger Institute (ESI), where part of the work was performed, for their hospitality. We thank Jorma Louko for a brief but fruitful exchange of emails that lead to the pointing of Ref. \cite{Louko:2006zv}.

\end{acknowledgments}

\appendix

\section{Evolution to a stationary state as $\lambda \to \infty$}
\label{app:Stat}

We have assumed throughout that $\widehat{\mathcal{W}}_\mathbb{C}$ is a polynomially bounded non-negative function. As discussed in Sec. \ref{sec:KMS}, this is a weak assumption due to the positivity of the state. For the purposes of the ensuing discussion, let us assume further that $\widehat{\mathcal{W}}_{\beta}^{+}(\omega) > 0$. In this case, let
\begin{subequations}
\label{AB}
\begin{align}
A(\omega) := & \frac{c^2 || \widehat{\chi} ||^2_{L^1}}{(2\pi)^2}\left( \widehat{\mathcal{W}}_{\beta}^{+}(- \omega) + \widehat{\mathcal{W}}_{\beta}^{-}(+ \omega) + \widehat{\mathcal{W}}_{\beta}^{+}(+ \omega) + \widehat{\mathcal{W}}_{\beta}^{-}(- \omega) \right), \label{A} \\
B(\omega) := & \frac{c^2 || \widehat{\chi} ||^2_{L^1}}{(2\pi)^2}\left( \widehat{\mathcal{W}}_{\beta}^{+}(+ \omega) + \widehat{\mathcal{W}}_{\beta}^{-}(+ \omega) \right) \label{B}.
\end{align}
\end{subequations}

By taking a time derivative on Eq. \eqref{mr11-infty} and Eq. \eqref{mr00-infty} and using the complex conjugate of Eq. \eqref{mr01-infty}, after some straightforward rearrangements, one obtains the following system of equations:
\begin{subequations}
\label{mrij-stat}
\begin{align}
\ddot{ \grr}^{(\infty)}_{11}(\grt) =& - A(\omega) \dot{\grr}^{(\infty)}_{11}(\grt), \label{mr11-stat} \\
\ddot{ \grr}^{(\infty)}_{00}(\grt) =& - A(\omega) \dot{\grr}^{(\infty)}_{00}(\grt), \label{mr00-stat} \\
\dot r_{01}(\tau) = &  -B(\omega) r_{01}(\tau) + B(-\omega) \left[ \cos(2 \omega \tau) r_{01}(\tau) + i_{01} \sin(2 \omega \tau) \right],  \label{mrr01-stat} \\
\dot i_{01}(\tau) = & -B(\omega) i_{01}(\tau) + B(-\omega) \left[ \cos(2 \omega \tau) i_{01}(\tau) - r_{01} \sin(2 \omega \tau) \right],  \label{mri01-stat}
\end{align}
\end{subequations}
where $r_{01} = \Re \left(\grr^{(\infty)}_{01}\right)$ and $i_{01} = \Im \left( \grr^{(\infty)}_{01}\right)$, which have as general solutions
\begin{subequations}
\label{mrij-sol}
\begin{align}
 \grr^{(\infty)}_{11}(\grt) = & \alpha_{11}(\omega) \exp\left[- A(\omega) \tau  \right] + \beta_{11}(\omega), \label{mr11-sol} \\
 \grr^{(\infty)}_{00}(\grt) = & \alpha_{00}(\omega) \exp\left[- A(\omega) \tau  \right] + \beta_{00}(\omega), \label{mr00-sol} \\
r_{01}(\tau) = &\gamma(\omega) \exp\left[-B(\omega)\tau+B(-\omega)\frac{\sin(2\omega\tau)}{2\omega}\right]\cos\left(\frac{B(-\omega)\cos(2\omega\tau)}{2\omega}\right) \nonumber\\ 
&-\delta(\omega) \exp\left[-B(\omega)\tau+B(-\omega)\frac{\sin(2\omega\tau)}{2\omega}\right]\sin\left(\frac{B(-\omega)\cos(2\omega\tau)}{2\omega}\right),
 \label{mr01-sol}\\
 i_{01}(\tau) = &\delta(\omega) \exp\left[-B(\omega)\tau+B(-\omega)\frac{\sin(2\omega\tau)}{2\omega}\right]\cos\left(\frac{B(-\omega)\cos(2\omega\tau)}{2\omega}\right) \nonumber\\ 
&+\gamma(\omega) \exp\left[-B(\omega)\tau+B(-\omega)\frac{\sin(2\omega\tau)}{2\omega}\right]\sin\left(\frac{B(-\omega)\cos(2\omega\tau)}{2\omega}\right),
 \label{mri01-sol}
\end{align}
\end{subequations}
where $\alpha_{00}, \beta_{00}, \alpha_{11}, \beta_{11}, \gamma, \delta$ are the general solution coefficients, and in principle functions of $\omega$. 

In the Born-Markov approximation, we are interested in $\tau/\tau_{\rm off} = O(\lambda)$ (where $\tau_{\rm off}$ is the switch-off time of $\chi$), so the off-diagonal elements of the density matrix become exponentially suppressed, cf. Eq. \eqref{mr01-sol} and \eqref{mri01-sol}, while the diagonal elements become constant, cf. Eq. \eqref{mr11-sol} and \eqref{mr00-sol}. This yields the condition
\be
\lim_{\lambda\to\infty}\frac{d\hat{\rho}^{(\lambda)} (\tau)}{d\tau}=0.
\ee



\begin{thebibliography}{}

\bibitem{Unruh} W. G. Unruh, ``Notes on black hole evaporation,'' \href{https://journals.aps.org/prd/abstract/10.1103/PhysRevD.14.870}{Phys. Rev. D {\bf 14}, 870 (1976)}.

\bibitem{Dewitt}
B. S. DeWitt, ``Quantum gravity: the new synthesis,” in {\it General Relativity:
an Einstein centenary survey} (edited by S. Hawking and W. Israel), 
Cambridge University Press, Cambridge, England, 1979.

\bibitem{Birrell}
N. D. Birrell and P. C. W. Davies, {\em Quantum Fields in Curved Space}, (Cambridge University Press, Cambridge, England, 1982).

\bibitem{hawking}
S. W. Hawking,
``Particle creation by black holes,''
\href{https://link.springer.com/article/10.1007\%2FBF01608497}{Commun. Math. Phys. {\bf 43}, 199 (1975)}.

\bibitem{Kubo} R. Kubo, ``Statistical mechanical theory of irreversible processes. I. General
theory  and  simple  applications  in  magnetic  and  conduction  problems,” \href{https://journals.jps.jp/doi/abs/10.1143/JPSJ.12.570}{
J. Phys. Soc. Jpn. {\bf 12}, 570 (1957)}.

\bibitem{Schwinger}  P. C. Martin and J. S. Schwinger, “Theory of many particle systems. I”, \href{https://journals.aps.org/pr/abstract/10.1103/PhysRev.115.1342}{Phys. Rev. {\bf 115}, 1342 (1959)}.

\bibitem{Haag}
R. Haag, {\sl Local Quantum Physics}, (Springer, New York, 1996).

\bibitem{Breuer}H. P. Breuer and F. P. Petruccione, {\em The Theory of Open Quantum Systems} (Oxford University Press, New York,  2007).

\bibitem{Letaw} J. R. Letaw, ``Stationary world lines and the vacuum excitation of noninertial detectors," \href{https://journals.aps.org/prd/abstract/10.1103/PhysRevD.23.1709}{Phys. Rev. D {\bf 23}, 1709 (1981)}.

\bibitem{BL}
C. J. Fewster, B. A. Ju\'arez-Aubry and J. Louko, ``Waiting for Unruh,'' \href{https://iopscience.iop.org/article/10.1088/0264-9381/33/16/165003/meta}{Classical Quantum Gravity {\bf 33}, 165003 (2016)}.

\bibitem{DMCA}
D. Moustos and C. Anastopoulos, ``Non-Markovian time evolution of an accelerated qubit," \href{https://journals.aps.org/prd/abstract/10.1103/PhysRevD.95.025020}{Phys. Rev. D {\bf 95}, 025020 (2017)}.

\bibitem{Louko:2006zv}
  J.~Louko and A.~Satz,
  ``How often does the Unruh-DeWitt detector click? Regularisation by a spatial profile,''
  \href{https://iopscience.iop.org/article/10.1088/0264-9381/23/22/015/meta}{Classical\ Quantum\ Gravity\  {\bf 23}, 6321 (2006)}.

\bibitem{Bell:Leinaas} J. S. Bell and J. M. Leinaas, ``Electrons as accelerated thermometers," \href{https://www.sciencedirect.com/science/article/pii/0550321383906016?via\%3Dihub}{Nucl. Phys. {\bf B212}, 131
(1983)}.

\bibitem{Unruh:circ} W. G. Unruh, ``Acceleration radiation for orbiting electrons'', \href{https://www.sciencedirect.com/science/article/abs/pii/S0370157398000684?via\%3Dihub}{Phys. Rept. {\bf 307}, 163 (1998)}.


\bibitem{DM}
D. Moustos, ``Asymptotic states of accelerated detectors and universality of the Unruh effect," \href{https://journals.aps.org/prd/abstract/10.1103/PhysRevD.98.065006}{Phys. Rev. D {\bf 98}, 065006 (2018)}.


\bibitem{DeBievre}
S. De Bievre and M. Merkli, ``The Unruh effect revisited'', \href{https://iopscience.iop.org/article/10.1088/0264-9381/23/22/026/meta}{Class. Quant. Grav. {\bf 23}, 6525 (2006)}.

\bibitem{Benatti} F. Benatti and R. Floreanini, ``Entanglement generation in uniformly accelerating atoms: Reexamination of the Unruh effect'', \href{https://journals.aps.org/pra/abstract/10.1103/PhysRevA.70.012112}{Phys. Rev. A {\bf 70}, 012112 (2004)}.

\bibitem{Kay:1988mu}
  B.~S.~Kay and R.~M.~Wald,
  ``Theorems on the Uniqueness and Thermal Properties of Stationary, Nonsingular, Quasifree States on Space-Times with a Bifurcate Killing Horizon'', \href{https://www.sciencedirect.com/science/article/pii/037015739190015E}{
  Phys.\ Rept.\  {\bf 207}, 49 (1991)}.

\bibitem{Sanders:2013vza}
  K.~Sanders,
  ``On the construction of Hartle-Hawking-Israel states across a static bifurcate Killing horizon'',
 \href{https://link.springer.com/article/10.1007%2Fs11005-015-0745-2}{Lett.\ Math.\ Phys.\  {\bf 105}, 575 (2015)}.



\bibitem{Good:2017xzl}
  M.~R.~R.~Good, T.~Oikonomou and G.~Akhmetzhanova, ``Uniformly accelerated point charge along a cusp'',
  \href{https://onlinelibrary.wiley.com/doi/abs/10.1002/asna.201713453}{Astron.\ Nachr.\  {\bf 338}, 1151 (2017).}
  
\bibitem{Rad:Singl} N. Rad and D. Singleton, ``A test of the circular Unruh effect using atomic electrons'', \href{https://link.springer.com/article/10.1140\%2Fepjd\%2Fe2012-30387-6}{
Eur. Phys. J. D {\bf 66}, 258 (2012)}.

\bibitem{LePf} J. R. Letaw and J. D. Pfautsch, ``Quantized scalar field in rotating coordinates'', \href{https://journals.aps.org/prd/abstract/10.1103/PhysRevD.22.1345}{Phys. Rev. D {\bf 22}, 1345 (1980)}.


\end{thebibliography}
\end{document}